\documentclass[pra,twocolumn, showpacs]{revtex4}
\pdfoutput=1
\usepackage[cp850]{inputenc}
\usepackage{amsmath,amsfonts,amssymb,amsthm}
\usepackage[english]{babel}
\usepackage{latexsym}
\usepackage{amscd}
\usepackage{float}
\usepackage{graphicx}
\usepackage{amsmath}
\usepackage{amsfonts}
\usepackage{amssymb}
\usepackage{color}
\usepackage{setspace}%
\renewcommand{\Pr}{\text{\rm Pr\,}}

\renewcommand{\ln}{\textrm{ln\,}}

\DeclareMathOperator{\JD}{JD}
\DeclareMathOperator{\QJD}{QJD}

\DeclareMathOperator{\R}{\mathbb{R}}

\DeclareMathOperator{\dT}{d_T}
\DeclareMathOperator{\CM}{CM}
\newtheorem{theorem}{Theorem}

\newtheorem{proposition}[theorem]{Proposition}
\newtheorem{cor}[theorem]{Corollary}
\newtheorem{lemma}{Lemma}

\newtheorem{definition}{Definition}
\renewcommand{\qed}{\hfill{\rule{2mm}{2mm}}}
\renewenvironment{proof}[1][]{\begin{trivlist}
\item[\hspace{\labelsep}{\bf\noindent Proof#1:\/}] }{\qed\end{trivlist}}
\newcommand{\beqn}{\begin{equation*}}
\newcommand{\eeqn}{\end{equation*}}
\newcommand{\beqr}{\begin{eqnarray}}
\newcommand{\eeqr}{\end{eqnarray}}
\newcommand{\beq}{\begin{equation}}
\newcommand{\eeq}{\end{equation}}
\newcommand{\beqrn}{\begin{eqnarray*}}
\newcommand{\eeqrn}{\end{eqnarray*}}

\def\ket#1{\mathinner{|{#1}\rangle}}

\newcommand{\ketbra}[2]{|#1\rangle \! \langle #2|}
\newcommand{\Tr}{\mbox{\rm Tr}}
\newcommand{\densop}[1]{\mathcal{B}_+^1(\mathcal{H}_{#1})}

\begin{document}
\title{Properties of Classical and Quantum Jensen-Shannon Divergence}
\author{Jop Bri\"et}
\email{jop.briet@cwi.nl.}
\author{Peter Harremo\"es}
\email{P.Harremoes@cwi.nl.}
\affiliation{Centrum Wiskunde \& Informatica, Science Park 123, 1098 XG Amsterdam, The Netherlands}
% \author{Flemming Tops\o e}
% \email{topsoe@math.ku.dk.}
% \affiliation{Department of Mathematical Sciences, University of Copenhagen,
% Universitetsparken 5 DK-2100 Copenhagen, Denmark}
\date{\today}

\begin{abstract}
\emph{Jensen-Shannon divergence} ($\JD$) is a symmetrized and smoothed version
of the most important divergence measure of information theory, Kullback
divergence. As opposed to Kullback divergence it determines in a very direct
way a metric; indeed, it is the square of a metric. We consider a family of
divergence measures ($\JD_{\alpha}$ for $\alpha>0$), 
the \emph{Jensen divergences of order $\alpha$}, which generalize $\JD$ as $\JD_{1}=\JD$. Using
a result of Schoenberg, we prove that $\JD_{\alpha}$ is the square of a
metric for $\alpha\in\left(  0,2\right]  ,$ and that the resulting metric
space of probability distributions can be isometrically embedded in a real
Hilbert space. \emph{Quantum Jensen-Shannon divergence} ($\QJD$) is a
symmetrized and smoothed version of quantum relative entropy and can be
extended to a family of \emph{quantum Jensen divergences of order
$\alpha$} ($\QJD_{\alpha}$). We strengthen results by Lamberti et al. by
proving that for qubits and pure states, $\QJD_{\alpha}^{1/2}$ 
is a metric space which can be isometrically embedded in a real Hilbert space when
$\alpha\in\left(  0,2\right]  .$ In analogy with Burbea and Rao's
generalization of JD, we also define \emph{general} QJD by associating a
Jensen-type quantity to any weighted family of states. Appropriate
interpretations of quantities introduced are discussed and bounds are derived
in terms of the total variation and trace distance.

\end{abstract}

\maketitle

%\doublespace

%\noindent{\bf Keywords}\; Jensen-Shannon divergence, Jensen divergence, total variation, quantum Jensen-Shannon divergence, Kullback divergence, concave-convex optimization.

\section{Introduction}

For two probability distributions $P=(p_{1},\dots,p_{n})$ and $Q=(q_{1}%
,\dots,q_{n})$ on a finite alphabet of size $n\geq 2,$ Jensen-Shannon divergence
($\JD$) is a measure of divergence between $P$ and $Q$. It measures the
deviation between the Shannon entropy of the mixture $(P+Q)/2$ and the mixture
of the entropies, and is given by
\begin{equation}
\JD(P,Q)=H\left(  \frac{P+Q}{2}\right)  -\frac{1}{2}\big(H(P)+H(Q)\big).
\label{JDdef}%
\end{equation}
Attractive features of this function are that it is everywhere defined,
bounded, symmetric and only vanishes when $P=Q$. Endres and Schindelin
\cite{EndresSch03} proved that it is the square of a metric, which we call the
\emph{transmission metric} ($\dT$). This result implies, for example, that
Banach's fixed point theorem holds for the space of probability distributions
endowed with the metric $\dT.$
A natural way to extend Jensen-Shannon divergence is to consider a mixture of
$k$ probability distributions $P_{1},\dots,P_{k}$, with weights $\pi_{1}%
,\dots,\pi_{k}$, respectively. With $\pi=(\pi_{1},\dots,\pi_{k})$, we can then
define the \emph{general Jensen divergence} as
\[
\JD^{\pi}(P_{1},\dots,P_{k})=H\left(  \sum_{i=1}^{k}\pi_{i}P_{i}\right)
-\sum_{i=1}^{k}\pi_{i}H(P_{i}).
\]
This was already considered by Gallager \cite{Gallager1968} in 1968, who
proved that, for fixed $\pi$, this is a convex function in $(P_1,\cdots,P_k)$.
Further identities and inequalities
were derived by Lin and Wong \cite{LinWong90,Lin91}, and
Tops\o e~\cite{Topsoe00ine}. It has found a variety of important applications:
Sibson~\cite{Sibson1969} showed that it has applications in biology and
cluster analysis, Wong and You \cite{WongYou} used it as a measure of distance
between random graphs, and recently, Rosso et al. used it to quantify the
deterministic vs. the stochastic part of a time series \cite{rosso2007}. For
its statistical applications we refer to El-Yaniv et al. \cite{Yanivetal} and
references therein.

Burbea and Rao \cite{Burbea1982} introduced another level of generalization,
based on more general entropy functions. For an interval $I$ in $\R$ and a
function $\phi:I\rightarrow\R$, they define the $\phi$-\emph{entropy} of $x\in
I^{n}$ (where $I^{n}$ denotes the Cartesian product of $n$ copies of $I$) as
\[
H_{\phi}(x)=-\sum_{i=1}^{n}\phi(x_{i}).
\]
Based on this, they define the \emph{generalized mutual information measure}
as
\[
\JD_{\phi}^{\pi}(P_{1},\dots,P_{k})=H_{\phi}\left(  \sum_{i=1}^{k}\pi_{i}%
P_{i}\right)  -\sum_{i=1}^{k}\pi_{i}H_{\phi}(P_{i}),
\]
for which they established some strong convexity properties. If $k=2$,
$I=[0,1]$ and $\phi$ is the function $x\rightarrow\frac{1}{\alpha-1}%
(x^{\alpha}-x)$, then $H_{\phi}$ defines the \emph{entropy of order
$\alpha.$} In this case, Burbea and Rao proved that $\JD_{\phi}^{\pi}$ is convex for all $\pi$, if and
only if $\alpha\in\lbrack1,2]$, except if $n=2$ when convexity holds if and only if
$\alpha\in\lbrack1,2]$ or $\alpha\in\lbrack3,11/3]$.

We focus on the functions $\JD_{\phi}^{\pi}$, where $k\geq2$, $I = [0,1]$ and
$\phi$ defines entropy of order~$\alpha.$ For ease of notation we
write these as $\JD_{\alpha}^{\pi}$ if $k\geq2$ and as $\JD_{\alpha}$ 
if $k=2$ and $\pi=(1/2,1/2)$.

Shannon entropy is \emph{additive} in the sense that the entropy of
independent random variables, defined as the entropy of their joint
distribution, is the sum of their individual entropies. Like Shannon entropy R{\'e}nyi of order $\alpha$ entropy is additive but in general R{\'e}nyi entropy is not convex \cite{Renyi1961}. The power entropy of order $\alpha$ is a monotone function of R{\'e}nyi entropy but, contrary to R{\'e}nyi entropy it is a concave function which is what we are interested in. The study of power entropy dates back to J.H. Havrda and F. Charvat \cite{Havrda1967}. Since then it was rediscovered independently several times \cite{Lindhard1971,Aczel1975,Tsallis88}, but we have chosen the more neutral term entropy of order $\alpha$ rather than calling it Hravda-Chervat-Lindhardt-Nielsen-Acz{\'e}l-Dar{'o}czy-Tsallis entropy. Entropy of order $\alpha$ is not
additive (unless $\alpha=1$). This is one of the reasons why
this function is used by physicists in attempts to model long range interaction in 
statistical mechanics, cf. Tsallis \cite{Tsallis88} and followers (can
be traced from a bibliography maintained by Tsallis). 

Martins et al.
\cite{martins:measures,martins:nonextensive,Martins2008,Martins2008b} give
non-extensive (i.e. non-additive) generalizations of $\JD$ based on entropies of order $\alpha$ and an
extension of the concept of convexity to what they call \emph{$q$-convexity}.
For these functions they extend Burbea and Rao's results in terms of $q$-convexity.

Distance measures between quantum states, which generalize probability
distributions, are of great interest to the field of quantum information
theory~\cite{Wootters1981,Braunstein1994,Lee2003,mlmp:qjsd04,lamberti:qjsd}.
They play a central role in state discrimination and in quantifying entanglement.
For example, the quantum relative entropy of two states $\rho_{1}$ and $\rho_{2}$, given by
$S(\rho_{1}\Vert\rho_{2})=-\mbox{\rm Tr}\rho_{1}(\ln\rho_{1}-\ln\rho_{2})$,
is a commonly used distance measure. (For a review of its basic properties and
applications see \cite{ScWe01}). However, it is not symmetric and does not
obey the triangle inequality. As an alternative, Lamberti et al.
\cite{majtey:qjsd,lamberti:qjsd, lamberti:naturalmetric} proposed to use the
(classical) $\JD$ as a distance function for quantum states, but also
introduced a quantum version based on the von Neumann entropy, which we denote
by $\QJD$. Like its classical variant, it is everywhere defined, bounded,
symmetric and zero only when the inputs are two identical quantum states. They
prove that it is a metric on the set of pure quantum states and that it is
close to the Wootter's distance and its generalization introduced by
Braunstein and Caves \cite{Braunstein1994}. Whether the metric property holds
in general is unknown.

As an analogue to $\JD_{\alpha}^{\pi}$ for quantum states, we introduce the
\emph{general quantum Jensen divergence of order $\alpha$}
($\QJD_{\alpha}^{\pi}$). In the limit $\alpha\to1$ we obtain the ``von Neumann
version'':
\[
\QJD^{\pi}(\rho_{1},\dots,\rho_{k}) =S\left(  \sum_{i=1}^{k}\pi_{i}\rho
_{i}\right)  -\sum_{i=1}^{k}\pi_{i}S(\rho_{i}),
\]
where $S(\rho) = -\mbox{\rm Tr}\rho\ln\rho$ is the von Neumann entropy. For
$k=2$ and $\pi= (1/2,1/2)$ one obtains the \emph{quantum Jensen
divergence of order $\alpha$} ($\QJD_{\alpha}$), which generalizes $\QJD$ as
$\lim_{\alpha\to1}\QJD_{\alpha} = \QJD$.

%which we discuss in the context of the classical capacity of a quantum channel, and for a quantum generalization of Huffman encoding called \emph{indeterminate length quantum encoding}, introduced by Schumacher and Westmoreland~\cite{ScWe01}.
%\subsubsection{Interpretations of the general Quantum JSD.}

\subsubsection{Our results.}

We extend the results of Endres and Schindelin, concerning the metric property
of $\JD,$ and those of Lamberti et al., concerning the metric property of
$\QJD,$ as follows:

\begin{itemize}
\item Denoting the set of probability distributions on a set $X$ by $M_{+}%
^{1}(X),$ we prove that for $\alpha\in(0,2],$ the pair $\left(
M_{+}^{1}(X),\JD_{\alpha}^{1/2}\right)  $ is a metric space which can be
isometrically embedded in a real separable Hilbert space.

\item Denoting the set of quantum states on qubits (2-dimensional Hilbert
spaces) by $\mathcal{B}_{+}^{1}(\mathcal{H}_{2})$ and the set of pure-states
on $d$-dimensional Hilbert spaces by $\mathcal{P}(\mathcal{H}_{d}),$ we prove
that for $\alpha\in\left(  0,2\right]  ,$ the pairs $\left(  \mathcal{B}%
_{+}^{1}(\mathcal{H}_{2}),\QJD_{\alpha}^{1/2}\right)  $ and $\left(
\mathcal{P}(\mathcal{H}_{d}),\QJD_{\alpha}^{1/2}\right)  $ are metric spaces
which can be isometrically embedded in a real separable Hilbert space.

\item We show that these results do \emph{not} extend to the cases $\alpha
\in(2,3)$ and $\alpha\in(\frac{7}{2},\infty)$. More precisely, we show that,
for $\alpha\in(2,3)$, neither $\JD_{\alpha}$ nor $\QJD_{\alpha}$ can be the
square of a metric, and for $\alpha\in(\frac{7}{2},\infty)$, isometric
embedding in a real Hilbert space is impossible (though the metric property
may still hold).
\end{itemize}

\subsubsection{Techniques.}

To prove our positive results, we evoke a theorem by Schoenberg which links
Hilbert space-embeddability of a metric space $(X,d)$ to the property of
\emph{negative definiteness} (defined in Section \ref{metricsec}). We
prove that for $\alpha\in\left(  0,2\right]  ,$ $\JD_{\alpha}$ satisfies this
condition for every set of probability distributions, and that
$\QJD_{\alpha}$ satisfies this condition for every set of qubits or pure-states.

%Though we concentrate on the discrete
%case, well-known limiting properties imply that the results related to
%classical divergence -- and this includes Jensen-Shannon divergence -- carry
%over to the general case of probability measures on arbitrary Borel spaces.

\subsection{Interpretations of $\JD^{\pi}$ and $\QJD^{\pi}$}

\subsubsection{Channel capacity.}

A \emph{discrete memoryless channel} is a system with input and output
alphabets $X$ and~$Y$ respectively, and conditional probabilities $p(y|x)$ for
the probability that $y\in Y$ is received when $x\in X$ is sent. For a
discrete memoryless channel with $|X|=k$, input distribution $\pi$ over $X$
and conditional distributions $P_{x}(y) = p(y|x)$, we have that $\JD^{\pi
}\big(P_{x_{1}},\dots,P_{x_{k}}\big)$ in fact gives the transmission rate.
(See for example \cite{CovThom}.) Inspired by this fact, we call the metric
defined by the square root of $\JD$ the \emph{transmission metric} and denote
it by $\dT$.

A \emph{quantum channel} has classical input alphabet $X,$ and an encoding of
every element $x\in X$ into a quantum state $\rho_{x}$. A receiver decodes a
message by performing a measurement with~$|Y|$ possible outcomes, on the state
he or she obtained. For a quantum channel with $|X|=k$, input distribution
$\pi$ over $X,$ and encoded elements~$\rho_{x},$ Holevo's Theorem
\cite{Holevo73} says that the maximum transmission rate of classical
information (the classical channel capacity) is at most $\QJD^{\pi}\left(
\rho_{x_{1}},\dots,\rho_{x_{k}}\right)  $. Holevo~\cite{Holevo1998}, and
Schumacher and Westmoreland~\cite{Schumacher1997} proved that this bound is
also asymptotically achievable.

\subsubsection{Data compression and side information.}

Let $X=[k]$ be an input alphabet and for each $i\in X$ let $P_{i}$ be a
distribution over output alphabet $Y$ with $|Y|=n$. Consider a setting where a
sender uses a weighting $\pi$ over $X$, and a receiver who has to compress the
received output data losslessly. We call the receiver's knowledge of which
distribution $P_{i}$ is used at any time the \emph{side information}, and
difference between the average number of nats (units based on the natural
logarithm instead of bits) used for the encoding when the side information is
known, and when it is not known, the \emph{redundancy}. In \cite{FugTop04ISIT}%
, this setting is referred to as the \emph{switching model}.

If the receiver always knows which input distribution is used, then for each
distribution $P_{i}$, he or she can apply the optimal compression encoding
given by $H(P_{i})$. Hence, if the receiver has access to the side
information, the average number of nats, that the optimal compression encoding
uses is given by $\sum_{i=1}^{k}\pi_{i}H(P_{i}).$

However, if the receiver does not know when which input distribution is used,
he or she always has to use the same encoding. 
We say that a compression encoding
$C$ \emph{corresponds to} an input distribution~$Q$, if $C$ is optimal for $Q$
(i.e., the number of nats used is $H(Q)$). If the sender transmits an infinite sequence of
letters $y_{1}y_{2}\cdots$ , picked according to distribution $P_{i},$ and the
receiver compresses it using an encoding $C$ which corresponds to distribution
$Q$, then the average number of used nats 
is given by $\sum_{j=1}^{n}P_{i}%
(y_{j})\ln\frac{1}{Q(y_{j})}$.

Hence, with the weighting $\pi_{1},\dots,\pi_{k}$, we get the redundancy
\begin{align*}
R(Q):=  &  \sum_{i=1}^{k}\left(  \pi_{i}H(P_{i})-\sum_{j=1}^{n}\pi_{i}%
P_{i}(y_{j})\ln\frac{1}{Q(y_{j})}\right) \\
&  =\sum_{i=1}^{k}\pi_{i}D(P_{i}\Vert Q)
\,,
\end{align*}
a weighted average of Kullback divergences between the $P_{i}$'s and $Q$. The
{\it compensation identity} states that for $\overline{P}=\sum_{i=1}^{k}\pi_{i}%
P_{i}$, the equality
\begin{equation}
\sum_{i=1}^{k}\pi_{i}D(P_{i}\Vert Q)=\sum_{i=1}^{k}\pi_{i}D(P_{i}%
\Vert\overline{P})+D(\overline{P}\Vert Q)\, \label{eq:6}%
\end{equation}
holds for any distribution 
$Q$, cf. \cite{Topsoe67,Topsoe01ide}.

  It follows immediately that
$Q=\overline{P}$ is the unique argmin-distribution for $R(Q)$, and that
$\JD^{\pi}(P_{1},\dots,P_{k})$ is the corresponding minimum value.

Analogously in a quantum setting, let $X=[k]$ be an input alphabet, and for
each $i\in X$ let $\rho_{i}$ be a state on an output Hilbert space
$\mathcal{H}_{Y}$. We can think of a sender who uses the weighting $\pi$ of
distributions $X$, but a receiver who has to compress the states on
$\mathcal{H}_{Y}$ using as few \emph{qubits} as possible.

Schumacher \cite{schumacher:qcoding} showed that the mean number of qubits
necessary to encode a state $\rho_{i}$ is given by~$S(\rho_{i})$. Later,
Schumacher and Westmoreland \cite{Schumacher2001} introduced a quantum
encoding scheme, in which an encoding~$C_{Q}$ that is optimal (i.e., requires
the least number of qubits) for a state $\sigma$ requires on average
$S(\rho_{i})+S(\rho_{i}\Vert\sigma)$ qubits to encode $\rho_{i}$. Hence, when
the receiver uses $C_{Q}$ as the encoding, the mean redundancy is
$R(\sigma):=\sum_{i=1}^{k}\pi_{i}S(\rho_{i}\Vert\sigma)$. Let $\bar{\rho}%
=\sum_{i=1}^{k}\pi_{i}\rho_{i}$. The quantum analogue of \eqref{eq:6} is given
by Donald's identity \cite{Donald87}:
\[
\sum_{i=1}^{k}\pi_{i}S(\rho_{i}\Vert\sigma)=\sum_{i=1}^{k}\pi_{i}S(\rho
_{i}\Vert\bar{\rho})+S(\bar{\rho}\Vert\sigma),
\]
from which it follows that $\sigma=\bar{\rho}$ is the argmin-state that the
receiver should code for, and that $\QJD^{\pi}(\rho_{1},\dots,\rho_{k})$ is
the minimum redundancy.

\section{Preliminaries and notation}

In this section we fix notation to be used throughout the paper. We also
provide a concise overview of those concepts from quantum theory which we
need. For an extensive introduction we refer to \cite{Nielsen2000}.

\subsection{Classical information theoretic quantities}

We write $[n]$ for the set $\{1,2,\dots,n\}$. The set of probability
distributions supported by $\mathbb{N}$ is denoted by $M_{+}^{1}(\mathbb{N})$
and the set supported by $[n]$ is denoted by $M_{+}^{1}(n)$. We associate with
probability distributions $P,Q\in M_{+}^{1}(n)$ point probabilities
$(p_{1},\dots,p_{n})$ and $(q_{1},\dots,q_{n})$, respectively. Entropy
of order $\alpha\neq1$, Shannon entropy and Kullback divergence are given by
\[
S_{\alpha}\left(  P\right)  :=\frac{1-\sum_{i=1}^{n}p_{i}^{\alpha}}{\alpha
-1},
\]%
\begin{equation}
H(P):=-\sum_{i=1}^{n}p_{i}\ln p_{i} \label{entropy}%
\end{equation}
and
\begin{equation}
D(P\Vert Q):=\sum_{i=1}^{n}p_{i}\ln\frac{p_{i}}{q_{i}}, \label{kulleib}%
\end{equation}
respectively. Note that
$\lim_{\alpha\rightarrow1^{+}}S_{\alpha}(P)=H(P)$. For two-point probability
distributions $P=(p,1-p)$ we let $s_{\alpha}(p)$ denote
$S_{\alpha}(p,1-p)$.

\subsection{Quantum theory}

\subsubsection{States.}

The $d$-dimensional complex Hilbert space, denoted by $\mathcal{H}_{d}$, is
the space composed of all $d$-dimensional complex vectors, endowed with the
standard inner product. A physical system is mathematically represented by a
Hilbert space. Our knowledge about a physical system is expressed by its
\emph{state}, which in turn is represented by a \emph{density matrix} (a
trace-1 positive matrix) acting on the Hilbert space. The set of
density matrices on a Hilbert space $\mathcal{H}$ is denoted by $\mathcal{B}%
_{+}^{1}\left(  \mathcal{H}\right)$ \footnote{This deviates from the notation
used in \cite{lamberti:qjsd}. We do this for the sake of consistency with
regard to the notation for probability distributions.}. Rank-1 density
matrices are called \emph{pure-states}. Systems described by two-dimensional
Hilbert spaces are called \emph{qubits}. As the eigenvalues of a density
matrix are always positive real numbers that sum to one, a state can be
interpreted as a probability distribution over pure-states. Hence, sets of
states with a complete set of common eigenvectors can be interpreted as
probability distributions on the same set of pure-states. States thus
generalize probability distributions. This interpretation is not possible when
a common basis does not exist. Two states $\rho$ and $\sigma$ have a set of
common eigenvectors if and only if they commute; i.e. $\rho\sigma=\sigma\rho$.

\subsubsection{Measurements.}

Information about a physical system can be obtained by performing a
\emph{measurement} on its state. The most general measurement with $k$
outcomes is described by $k$ positive matrices $A_{1},\dots
,A_{k},$ which satisfy $\sum_{i=1}^{k}A_{i}=I$. This is a special case of the
more general concept of a \emph{positive operator valued measure} (POVM, see
for example \cite{Nielsen2000}). The probability that a measurement $A$ of a
system in state $\rho$ yields the $i$'th outcome is $\mbox{\rm Tr}(A_{i}\rho
)$. Hence, the measurement yields a random variable $A(\rho)$ with $\Pr
[A(\rho)=\lambda_{i}]=\mbox{\rm Tr}(A_{i}\rho)$. Naturally, the measurement
operators and quantum states should act on the same Hilbert space.

\subsection{Quantum information theoretic quantities}

\label{qinfoquants} For states $\rho,\sigma\in\mathcal{B}_{+}^{1}\left(
\mathcal{H}\right)  $, we use the quantum version of entropy of order $\alpha,$ von
Neumann entropy and quantum relative entropy, given by
\[
S_{\alpha}\left(  \rho\right)  :=\frac{1-\mbox{\rm Tr}\left(  \rho^{\alpha
}\right)  }{\alpha-1},
\]%
\begin{equation}
S(\rho):=-\mbox{\rm Tr}(\rho\ln\rho)\label{vnent}%
\end{equation}
and
\begin{equation}
S(\rho\Vert\sigma):=\mbox{\rm Tr}\rho\ln\rho-\mbox{\rm Tr}\rho\ln
\sigma,\label{relent}%
\end{equation}
respectively. Note that $\lim_{\alpha\rightarrow1^{+}}S_{\alpha}(\rho
)=S(\rho).$ We refer to \cite{Ohya1993} for a discussionof quantum relative entropy.

\section{Divergence measures}

\label{divmeasures:section}

\subsection{The general Jensen divergence}

Let us consider a mixture of $k$ probability distributions $P_{1},\dots,P_{k}$ with
weights $\pi_{1},\dots,\pi_{k}$ and let $\overline{P} = \sum_{i=1}^{k}\pi
_{i}P_{i}$. Jensen's inequality and concavity of Shannon entropy implies that
\[
H\left(  \sum_{i=1}^{k}\pi_{i}P_{i}\right)  \geq\sum_{i=1}^{k}\pi_{i}%
H(P_{i}).
\]
When entropies are finite, we can subtract the right-hand side from the
left-hand side and use this as a measure of how much Shannon entropy deviates
from being affine. This difference is called the \emph{general Jensen-Shannon
divergence} and we denote it by $\JD^{\pi}(P_{1},\dots,P_{k})$, where $\pi=
(\pi_{1},\dots,\pi_{k})$. One finds that
\begin{equation}
H\left(  \sum_{i-1}^{k}\pi_{i}P_{i}\right)  -\sum_{i=1}^{k}\pi_{i}
H(P_{i})=\sum_{i=1}^{k}\pi_{i}D(P_{i}\Vert\overline{P})\, \label{cjsd}%
\end{equation}
and therefore
\begin{equation}
\label{jsd}\JD^{\pi}(P_{1},\dots,P_{k}) =\sum_{i=1}^{k}\pi_{i}D(P_{i}%
\Vert\overline{P})\,.
\end{equation}
In the general case when entropies may be infinite the last expression can be
used, but we will focus on the situation where the distributions are over a
finite set and in this case we can use the left-hand side of (\ref{cjsd}).

Jensen divergence of order $\alpha$ is defined by the formula
\[
\JD_{\alpha}^{\pi}(P_{1},\dots,P_{k}) =S_{\alpha}\left(  \sum_{i=1}^{k}\pi
_{i}P_{i}\right)  -\sum_{i=1}^{k}\pi_{i}S_{\alpha} (P_{i}).
\]

Similarly, if $\rho_{1},\dots,\rho_{k}$ are states on a Hilbert space we
define
\begin{equation}
\QJD^{\pi}(\rho_{1},\dots,\rho_{k}) =\sum_{i=1}^{k}\pi_{i}S(\rho_{i}%
\Vert\overline{\rho}), \label{qjsd:defrel}%
\end{equation}
where $\overline{\rho}=\sum_{i=1}^{k}\pi_{i}\rho_{i}$. For states on a finite
dimensional Hilbert space we have
\[
\QJD^{\pi}(\rho_{1},\dots,\rho_{k}) =S\left(  \sum_{i=1}^{k}\pi_{i}\rho
_{i}\right)  -\sum_{i=1}^{k}\pi_{i}S(\rho_{i}).
\]
The quantum Jensen divergence of order $\alpha$ is defined by
\[
\QJD_{\alpha}^{\pi}(\rho_{1},\dots,\rho_{k}) =S_{\alpha}\left(  \sum_{i=1}%
^{k}\pi_{i}\rho_{i}\right)  -\sum_{i=1}^{k}\pi_{i}S_{\alpha}(\rho_{i}).
\]

\subsection{The Jensen divergence}

For even mixtures of two distributions, we introduce the notation
$\JD_{\alpha}(P,Q)$ for $\JD_{\alpha}(\frac{1}{2}P+\frac{1}{2}Q)$. That is,
\begin{equation}
\label{eq:8}\JD_{\alpha}(P,Q) := S_{\alpha}\left(  \frac{P+Q}{2}\right)  -
\frac{1}{2}S_{\alpha}(P) - \frac{1}{2}S_{\alpha}(Q).
\end{equation}
For even mixtures of two states the $\QJD$ was defined in~\cite{majtey:qjsd},
to which we refer for some of its basic properties. We consider the order $\alpha$
version of this and write $\QJD_{\alpha}(\rho,\sigma)$ for $\QJD_{\alpha
}(\frac{1}{2}\rho+ \frac{1}{2}\sigma)$. That is,
\begin{equation}
\label{qjsd}\QJD_{\alpha}(\rho,\sigma) := S_{\alpha}\left(  \frac{\rho+\sigma
}{2}\right)  - \frac{1}{2}S_{\alpha}(\rho) - \frac{1}{2}S_{\alpha}(\sigma).
\end{equation}

We refer to (\ref{eq:8}) and (\ref{qjsd}) simply as \textit{Jensen
divergence of order $\alpha$} ($\JD_{\alpha}$) and \emph{quantum
Jensen divergence of order $\alpha$} ($\QJD_{\alpha}$) respectively.

\section{Metric properties}

\label{metricsec}

\label{metricprops}

In this section we borrow most of the notational conventions and definitions
from Deza and Laurent \cite{Deza1997}. We refer to this book, to Berg, Christensen and Ressel \cite{Bergetal84}, and to Blumenthal
\cite{Blumenthal1953} for extensive introductions to the used results. Like Berg, Christensen and Ressel \cite{Bergetal84} we shall use the expressions ``positive and negative definite'' for what most textbook would call ``positive and negative semi-definite''.

\begin{definition}
\label{MetricDef}For a set $X$, a function $d:X\times X\rightarrow\mathbb{R}$
is called a \emph{distance} if for every $x,y\in X$:

\begin{enumerate}
\item $d(x,y)\geq0$ with equality if $x=y$.

\item $d$ is symmetric: $d(x,y)=d(y,x)$.

The pair $(X,d)$ is then called a \emph{distance space}. If in addition to 1
and 2, for every triple $x,y,z\in X$, the function $d$ satisfies
\end{enumerate}
\end{definition}

\begin{itemize}
\item[3.] $d(x,y)+d(x,z)\geq d(y,z)$ (the triangle inequality),

then $d$ is called a \emph{pseudometric} and $(X,d)$ a \emph{pseudometric
space}. If also, $d(x,y)=0$ holds if \emph{and only if} $x=y,$ then we speak
of a \emph{metric} and a \emph{metric space}.
\end{itemize}

Our techniques to prove our embeddability results for $\JD_{\alpha}$ and
$\QJD_{\alpha}$ are somewhat indirect. To provide some intuition, we briefly
mention the following facts. Only Definition \ref{MetricDef}, Proposition \ref{mengerprop} and Theorem
\ref{theorem:1} are needed for our proofs.

Work of Cayley and Menger gives a characterization of $\ell_{2}$ embeddability
of a distance space in terms of Cayley-Menger determinants. Given a finite
distance space $(X,d)$, the Cayley-Menger matrix $\CM(X,d)$ is given in terms
of the matrix $D_{ij} = d(x_{i},x_{j})$, for $x_{i},x_{j}\in X$, and the
all-ones vector~$e$:
\begin{align*}
\CM(X,d):=%
\begin{pmatrix}
D & e\\
e^{T} & 0
\end{pmatrix}
.
\end{align*}

Menger proved the following relation between $\ell_{2}$ embeddability and the
determinant of $\CM(X,d)$.

\begin{proposition}
[\cite{Menger1954}]\label{mengerprop} Let $(X,d)$ be a finite distance space.
Then $(X,d^{1/2})$ is $\ell_{2}$ embeddable if and only if for every
$Y\subseteq X$, we have $(-1)^{|Y|}\det\CM(Y,d)\geq0$.
\end{proposition}

As an example, consider a distance space with $|X|=3$. If we set
$a:=d(x_{1},x_{2})^{1/2}$, $b:=d(x_{1},x_{3})^{1/2}$ and $d(x_{2},x_{3}%
)^{1/2}$, then we obtain
\begin{multline}
-\det\CM(X,d)=\label{CMHeron}\\
(a+b+c)(a-b-c)(-a+b-c)(-a-b+c).%\nonumber
\end{multline}
On the one hand, this at least zero if $d$ is a pseudometric, and hence
pseudometric spaces on three points are $\ell_{2}$ embeddable.
On the other
hand, up to a factor $1/16$, the right-hand-side of \eqref{CMHeron} is the
square of Heron's formula for the area of a triangle with edge-lengths $a$,
$b$ and $c$. In general, Cayley-Menger determinants give the formulas needed
to calculate the squared hypervolumes of higher dimensional simplices.
Menger's result can thus be interpreted as saying that a distance space
$(X,d^{1/2})$ is $\ell_{2}$ embeddable if and only if every subset is a
simplex with real hypervolume.

Returning to our example with $|X|=3$, we also have the following implication.

\begin{proposition}
\label{schoenheronprop} \label{pro:4} Let $(\{x_{1},x_{1},x_{3}\},d)$ be a
distance space. Assume that for every $c_{1},c_{2},c_{3}\in\mathbb{R}$ such
that $c_{1}+c_{2}+c_{3}=0$, the distance function $d$ satisfies
\begin{equation}
\sum_{i,j}c_{i}c_{j}d(x_{i},x_{j})\leq0,\label{threensd}%
\end{equation}
where the summation is over all pairs $i,j\in\{1,2,3\}$. Then $(\{x_{1}%
,x_{1},x_{3}\},d^{1/2})$ is $\ell_{2}$ embeddable.
\end{proposition}

\begin{proof}
Let $a:=d(x_1,x_2)^{1/2}$, $b:=d(x_1,x_3)^{1/2}$ and $c:=d(x_2,x_3)^{1/2}$.
%If $d$ is indeed the square of a metric, then by the triangle inequality we have $a+b-c\geq 0$, $a-b+c\geq 0$ and $-a+b+c\geq 0$. Notice that two of these inequalities are always automatically satisfied for any triple $(a,b,c)$ of real numbers.
%Heron's formula says that a triangle with edges of lengths $a'$, $b'$ and $c'$, has an area $A$ satisfying
%\beqr
%16A^2 &=& (a'+b' + c')(a'+b' - c')(a'-b' + c')(-a'+b' + c')\nonumber\\
%&=& -a'^2 - b'^2 - c'^2 + 2a'b' + 2a'c' + 2b'c'\label{Heron}
%\eeqr
%Hence, if we interpret $a$, $b$ and $c$ as the lengths of the edges of a triangle and get that Heron's formula gives a non-imaginary area, then $K^{1/2}$ satisfies the triangle inequality.
We first show that \eqref{threensd} implies that~\eqref{CMHeron} is nonnegative. To this end, set $c_{1}=1$, $c_{2}=t$,
$c_{3}=-t-1$ where $t$ is a real parameter. Then, if~\eqref{threensd} holds, we get the inequality
\[
a^2t+ b^2t(-t-1)+c^2(-t-1)\leq0\,.
\]
The nonnegativity of \eqref{CMHeron} follows from the fact that this inequality holds if and only if the discriminant of this second order polynomial is at least zero. The result now follows from Proposition \ref{mengerprop}.
\end{proof}

%% First note that a triple $(\alpha,\beta,\gamma)$ of non-negative
%% numbers are \textit{squared distances} if $\alpha^{1/2}\leq\beta^{{1/2}%
%% }+\gamma^{{1/2}}$, $\beta^{{1/2}}\leq\gamma^{{1/2}}+\alpha^{{1/2}}$ and
%% $\gamma^{{1/2}}\leq\alpha^{{1/2}}+\beta^{{1/2}}$. Assuming that $\gamma
%% =\max\{\alpha,\beta,\gamma\}$, this amounts to the condition $\gamma^{{1/2}%
%% }\leq\alpha^{{1/2}}+\beta^{{1/2}}$ which, after some simple algebra can be
%% transformed to the equivalent inequality
%% \begin{equation}
%% \alpha^{2}+\beta^{2}+\gamma^{2}\leq2\alpha\beta+2\beta\gamma+2\gamma
%% \alpha\,.\label{eq:15}%
%% \end{equation}

%This example serves as motivation for the following definition.

The basis of our positive results in this section is that, due to Schoenberg
\cite{Schoenberg1935,Schoen38}, a more general version of Proposition
\ref{schoenheronprop} also holds. To state it concisely, we first define
\emph{negative definiteness}.

\begin{definition}
[Negative definiteness]Let $(X,d)$ be a distance space. Then $d$ is said
to be \emph{negative defninite} if and only if for all finite sets
$(c_{i})_{i\leq n}$ of real numbers such that $\sum_{i=1}^{n}c_{i}=0$, and all
corresponding finite sets $(x_{i})_{i\leq n}$ of points in $X$, it holds that
\begin{equation}
\sum_{i,j}c_{i}c_{j}d(x_{i},x_{j})\leq0. \label{eq:12}%
\end{equation}
In this case, $(X,d)$ is said to be a distance space of \emph{negative type}.
\end{definition}

The following theorem follows as a corollary of Schoenberg's theorem.

\begin{theorem}
\label{theorem:1} Let $(X,d)$ be a distance space. Then $\left(
X,d^{1/2}\right)  $ can be isometrically embedded in a real separable Hilbert
space if and only if $(X,d)$ is of negative type.
\end{theorem}

%Without recourse to the rather deep theorem of Schoenberg, one can establish
%the central metric property directly from negative definiteness. This is
%convenient especially for readers with a basic interest in applications as the
%metric property is the key fact of interest. In order to short-cut the theorem
%of Schoenberg, one needs a simple observation: If one relaxes the basic
%requirement of negative definiteness, only considering triples of $c^{\prime}%
%$s and $x^{\prime}$s, the metric property still follows:

Note that if isometric embedding in a Hilbert space is possible, then the
space must be a metric space. We define \emph{positive definiteness} as follows.

\begin{definition}
[Positive definiteness]Let $X$ be a set and $f:X\times X\rightarrow
\mathbb{R}$ a mapping. Then $f$ is said to be \emph{positive definite} if
and only if for all finite sets $(c_{i})_{i\leq n}$ of real numbers and all
corresponding finite sets $(x_{i})_{i\leq n}$ of points in $X$, it holds that
\begin{equation}
\sum_{i,j}c_{i}c_{j}f(x_{i},x_{j})\geq0. \label{psddef}%
\end{equation}

\end{definition}

Because we are concerned with functions defined on convex sets, the following
definition shall be useful.

\begin{definition}
[Exponential convexity]Let $X$ be a convex set and $\phi:X\rightarrow
\mathbb{R}$ a mapping. Then $\phi$ is said to be \emph{exponentially convex}
if the function $X\times X\rightarrow\mathbb{R}$ given by $\left(  x,y\right)
\rightarrow\phi\left(  \frac{x+y}{2}\right)  $ is positive definite.
\end{definition}

Normally exponential convexity is defined as positive definiteness of $\phi\left(  x+y\right)$ (as is done in for instance \cite{Nussbaum1972}), but the definition given here allows the function $\phi$ only to be defined on a convex set.

\subsection{Metric properties of $\JD_{\alpha}$}

\label{jsdmetric}

With Theorem \ref{theorem:1} we prove the following for Jensen
divergence of order $\alpha.$

\begin{theorem}
\label{theorem:2} For $\alpha\in(0,2],$ the space $\left(  M_{+}%
^{1}(\mathbb{N}),\JD_{\alpha}^{1/2}\right)  $ can be isometrically embedded in
a real separable Hilbert space.
%Let $\alpha\in(0,2]$. Then, there exists a bijection $\Phi$ between $M_{+}^{1}(\mathbb{N})$ and a subset $\mathcal{H}_{\JD_{\alpha}}$ of a real separable Hilbert space $\mathcal{H}$ with norm $\|\cdot\|$, such that for all
%$(P,Q)\in M_{+}^{1}(\mathbb{N})\times M_{+}^{1}(\mathbb{N})$, we have that ${\JD_{\alpha
%}(P,Q)=\Vert\Phi(P)-\Phi(Q)\Vert^{2}}$.

\end{theorem}

Note that Theorem \ref{theorem:2} implies that the same holds for
$\QJD_{\alpha}$ for sets of commuting quantum states.

We use the following lemma to prove that $\JD_{\alpha}$ is negative
definite for $\alpha\in(0,2].$ Theorem \ref{theorem:2} then follows
from this and Theorem \ref{theorem:1}.

%The following lemma establishes negative definiteness of $\JD_{\alpha}$ for $\alpha\in(0,2]$. Theorem \ref{theorem:2} therefore follows from Schoenberg's theorem.
%The lemma is proved in a basically elementary way, using facts and
%tricks from harmonic analysis (see for example \cite{Bergetal84}).

\begin{lemma}
\label{powerrep} For $\alpha\in(0,1),$ we have
\[
x^{\alpha}=\frac{1}{\Gamma(-\alpha)}\int_{0}^{\infty}\frac{e^{-xt}%
-1}{t^{\alpha+1}}dt,
\]
where $\Gamma(\alpha)=\int_{0}^{\infty}t^{\alpha-1}e^{-t}dt$ is the Gamma
function. For $\alpha\in(1,2),$ we have
\[
x^{\alpha}=\frac{1}{\Gamma(-\alpha)}\int_{0}^{\infty}\frac{e^{-xt}%
-(1-xt)}{t^{\alpha+1}}dt.
\]

\end{lemma}

\begin{proof}
Let $\gamma\in(-1,0)$. From the definition of the Gamma function, we have the following equality:
\beqn
z^{\gamma} = z^{\gamma}\frac{1}{\Gamma(-\gamma)}\int_0^{\infty}r^{-(\gamma+1)}e^{-r}dr.
\eeqn
By substituting $r=tz$ we get
\beqn
z^{\gamma} =  \frac{1}{\Gamma(-\gamma)}\int_0^{\infty}\frac{e^{-zt}}{t^{\gamma+1}}dt.
\eeqn
Let $\beta\in(0,1)$ such that $\beta = \gamma+1$. Integrating $z^{\gamma}$ for $z$ from zero to $y$ and multiplying by $\gamma+1$ gives,
\beqrn
y^{\beta} &=& (\gamma+1)\int_0^yz^{\gamma}dz \\
&=& \frac{1}{\Gamma(-\beta)}\int_0^{\infty}\frac{e^{-yt}-1}{t^{\beta+1}}dt.
\eeqrn
Now let $\alpha\in(1,2)$ such that $\alpha = \beta+1$. Integrating $y^{\beta}$ and multiplying by $\beta+1$ gives the result.
\beqrn
x^{\alpha} &=& (\beta+1)\int_0^xy^{\beta}dy \\
&=& \frac{1}{\Gamma(-\alpha)}\int_0^{\infty}\frac{e^{-xt} -(1-xt)}{t^{\alpha+1}}dt.\qedhere
\eeqrn
\end{proof}

\begin{lemma}
\label{jsdnsdlem} For $\alpha\in\left(  0,2\right]  $, the distance space
$\left(  M_{+}^{1},\JD_{\alpha}\right)  $ is of negative type.
\end{lemma}

\begin{proof}
Let $(c_i)_{i\leq n}$ be a set of real numbers such that $\sum_{i=1}^nc_i=0$. For two probability
distributions $P$ and $Q$, we have
\[
\JD_{\alpha}(P,Q)=S_{\alpha}\left(  \frac{P+Q}{2}\right)  -\frac{1}%
{2}S_{\alpha}(P)-\frac{1}{2}S_{\alpha}(Q).
\]
Observe that for any real valued, single-variable function $f$, we have $\sum_{i,j}c_{i}c_{j}f(x_{i})=0$.
Hence, we only need to prove that the function
\[
S_{\alpha}\left(  \frac{P+Q}{2}\right)  =\frac{1}{\alpha-1}-\frac{1}%
{(\alpha-1)}\sum_{i}\left(  \frac{p_{i}+q_{i}}{2}\right)  ^{\alpha}%
\]
is negative definite for all $\alpha\in(0,2]$. From this decomposition of $S_{\alpha}$ into a sum over point
probabilities it follows that we need to show that $x\curvearrowright
x^{\alpha}$ is exponentially convex. Lemma \ref{powerrep} shows that for fixed $0<\alpha<1$ and fixed $1<\alpha<2$,
the mapping $x\curvearrowright-x^{\alpha}$ can be obtained as the limit of linear
combinations with positive coefficients of functions of the type
$x\curvearrowright1-e^{-tx}$ and $x\curvearrowright1-e^{-tx} - tx$ respectively. Each such function is exponentially convex
since the linear terms are, and for  non-negative real numbers $x_{1},\dots
,x_{n}$,
\[
\sum_{i,j}c_{i}c_{j}(-e^{-t(x_{i}+x_{j})})=-\left(  \sum_{i=1}^{n}%
c_{i}e^{-tx_{i}}\right)  ^{2}\leq0.
\]
The case $\alpha=1$ follows by continuity. The case $\alpha=2$ also follows by
continuity, but a direct proof without Lemma~\ref{powerrep}	 is straightforward.
\end{proof}

\begin{proof}[ of Theorem \ref{theorem:2}]
Follows directly from Lemma \ref{jsdnsdlem} and Theorem \ref{theorem:1}.
\end{proof}

A constructive proof of Theorem \ref{theorem:2} for $\JD_{1}$ ($\JD$) is given
by Fuglede \cite{Fuglede03,Fuglede2004}, who uses an embedding into a subset of a real
Hilbert space defined by a logarithmic spiral.

\subsection{Metric properties of $\QJD_{\alpha}$ for qubits}

Using the same approach as above, we prove the following for quantum
Jensen divergence of order~$\alpha$ and states on two-dimensional
Hilbert spaces.

\begin{theorem}
\label{qjsd:qubits} For $\alpha\in(0,2]$, the space 

\[\left(
\mathcal{B}_{+}^{1}(\mathcal{H}_{2}),\QJD_{\alpha}^{1/2}\right)  
\]
can be
isometrically embedded in a real separable Hilbert space.
%Let $\alpha\in(0,2]$. Then there exists a bijection $\Phi$ between $\densop{2}$ and a subset $\mathcal{H}_{\QJD_{\alpha}}$ of a real separable Hilbert space $\mathcal{H}$ with norm $\|\cdot\|$, such that for all
%$(\rho,\sigma)\in \densop{2}\times\densop{2}$, we have that ${\QJD_{\alpha
%}(\rho,\sigma)=\Vert\Phi(\rho)-\Phi(\sigma)\Vert^{2}}$.

\end{theorem}

%\begin{theorem}
%\label{qjsd:qubits} The function $\QJD_{\alpha}$ is negative definite on
%$\mathcal{B}_{+}^{1}(\mathcal{H}_{2})\times\mathcal{B}_{+}^{1}(\mathcal{H}_{2})$. Therefore, there exists a subset $\overline{\mathcal{H}}_{\QJD}\subseteq\overline{\mathcal{H}}$ of a real separable Hilbert space $\overline{\mathcal{H}}$ and a one-to-one bijection $\Phi$ between
%$\mathcal{B}_{+}^{1}(\mathcal{H}_{2})$ and $\overline{\mathcal{H}}_{\QJD}$
%such that, for all $(\rho,\sigma)\in\mathcal{B}_{+}^{1}(\mathcal{H}_{2}%
%)\times\mathcal{B}_{+}^{1}(\mathcal{H}_{2})$, $\QJD_{\alpha}(\rho
%,\sigma)=\Vert\Phi(\rho)-\Phi(\sigma)\Vert^{2}$ with $\Vert\cdot\Vert$
%denoting the norm in $\overline{\mathcal{H}}$.
%\end{theorem}

This is established by the following lemmas and Theorem \ref{theorem:1}.

\begin{lemma}
\label{Hmetricnsd} Let $\left(  V,\langle\cdot|\cdot\rangle\right)  $ be a
real Hilbert space with norm $\Vert\cdot\Vert_{2}=\langle\cdot|\cdot
\rangle^{1/2}$. Then, $\left(  V,\Vert\cdot\Vert_{2}^{2}\right)  $ is a
distance space of negative type.
\end{lemma}

\begin{proof}
The result follows immidiately if we expand the distance function $\|\cdot\|_2^2$ in terms of the inner product:
\begin{multline*}
\sum_{i,j}c_ic_j\langle x_i-x_j,x_i-x_j\rangle\\
= \sum_{i,j}c_ic_j\big(\|x_i\|_2^2 + \|x_j\|_2^2 - 2\langle x_i,x_j\rangle\big)\\
= 2\sum_ic_i\sum_jc_j\|x_j\|_2^2 - 2\sum_{i,j}c_ic_j\langle x_i,x_j\rangle\\
= 0 - 2\sum_{i,j}c_ic_j\langle x_i,x_j\rangle\\
= -2\Big\|\sum_ic_ix_i\Big\|_2^2 \leq 0.\qedhere
\end{multline*}
\end{proof}

\begin{lemma}
\label{qjsdnsdlem} The distance space $\left(
\mathcal{B}_{+}^{1}(\mathcal{H}_{2}),\QJD_{\alpha}\right) , \alpha\in(0,2]$ is of negative type.
\end{lemma}

\begin{proof}
Using the same techniques as in the proof of Theorem \ref{theorem:2}, and the fact that Lemma \ref{powerrep} also holds when~$x$ is a matrix, what has to be shown is that for $\rho\in\densop{2}$, the function $\rho\curvearrowright\mbox{\rm Tr}\left(  \exp\left(  -t\rho\right)
\right)  $ is exponentially convex.
Since $\rho$ acts on a two-dimensional Hilbert space, it has only two eigenvalues, $\lambda_{+}$ and $\lambda_{-}$, that satisfy $\lambda_{+}+\lambda_{-}   =1$ and $\lambda_{+}^{2}+\lambda_{-}^{2}   =\mbox{\rm Tr}\left(  \rho^{2}\right)$. A straightforward calculation gives
\begin{equation}
\lambda_{+/-}=\frac{1}{2}\pm\frac{\left(  2\mbox{\rm Tr}\left(  \rho
^{2}\right)  -1\right)  ^{1/2}}{2}.\label{qeval}%
\end{equation}
Plugging this into $\mbox{\rm Tr}\left(\exp(-t\rho)\right)  $ gives
\begin{align*}
\mbox{\rm Tr}\left(  e^{-t\rho}\right)    & =2e^{-t/2}\cosh\left(  \frac{t}%
{2}\left(  2\mbox{\rm Tr}\left(  \rho^{2}\right)  -1\right)  ^{1/2}\right)
\\
& =2e^{-t/2}\sum_{k=0}^{\infty}\frac{t^{2k}}{(2k)!4^{k}}\left(  2\mbox{\rm Tr}\left(
\rho^{2}\right)  -1\right)  ^{k},
\end{align*}
where the second equality follows form the Taylor expansion of hyperbolic
cosine. The task can thus be reduced to proving that $\left(
2\mbox{\rm Tr}\left(  \rho^{2}\right)  -1\right)  ^{k}$ is exponentially
convex for all $k\geq0$. For this we can use the following theorem:
\begin{theorem}
[{\cite[Slight reformulation of  Theorem 1.12]{Bergetal84}}%
]\label{productfunctions} Let $\phi_{1},\phi_{2}:X\curvearrowright\mathbb{C}$ be
exponentially convex functions. Then $\phi_{1}\cdot\phi_{2}$ is
exponentially convex too.
\end{theorem}
This implies that proving it for $k=1$ suffices. The trace distance of two
density matrices is defined as the Hilbert-Schmidt norm $\Vert\cdot\Vert_{2}$ of
their difference. Since the Hilbert-Schmidt norm is a Hilbert-space metric,
Lemma \ref{Hmetricnsd} implies that $(\rho_1,\rho_2)\curvearrowright\Vert\rho_1-\rho_2
\Vert_{1}^{2}$ is negative definite and the equality
\[
\Vert\rho_1-\rho_2\Vert_{2}^{2}=\mbox{\rm Tr}(\rho_1-\rho_2)^{2}%
=2(\mbox{\rm Tr}\rho_1^{2}+\mbox{\rm Tr}\rho_2^{2})-\mbox{\rm Tr}\big((\rho_1
+\rho_2)^{2}\big)%
\]
implies that the function $\mbox{\rm Tr}\left(  (\rho_1+\rho_2)^{2}\right)  $ is
positive definite. From this it follows that the function $2\mbox{\rm Tr}\left(  \rho^{2}\right)  -1$ is exponentially
convex.
\end{proof}

\begin{proof}[ of Theorem \ref{qjsd:qubits}]
Follows directly from Lemma \ref{qjsdnsdlem} and Theorem \ref{theorem:1}.
\end{proof}

\subsection{Metric properties of $\QJD_{\alpha}$ for pure-states}

Here we prove that $\QJD_{\alpha}$ is the square of a metric when restricted
to pairs of pure-states. For a Hilbert space of dimension $d$ we denote the
set of pure-states as $P(\mathcal{H}_{d})$.

\begin{theorem}
\label{qjsd:pure} For $\alpha\in(0,2]$, the space $\big(P(\mathcal{H}%
_{d}),\QJD_{\alpha}^{1/2}\big)$ can be isometrically embedded in a real
separable Hilbert space.
%Let $\alpha\in(0,2]$. Then, there exists a bijection $\Phi$ between $P(\mathcal{H}_d)$ and a subset $\mathcal{H}_{\QJD_{\alpha}}$ of a real separable Hilbert space $\mathcal{H}$ with norm $\|\cdot\|$, such that for all
%$(\rho,\sigma)\in P(\mathcal{H}_d)\times P(\mathcal{H}_d)$, we have that ${\QJD_{\alpha
%}(\rho,\sigma)=\Vert\Phi(\rho)-\Phi(\sigma)\Vert^{2}}$.

\end{theorem}

%\begin{theorem}
%\label{qjsd:pure} For any $d\geq 1$, the function $\QJD$ is negative
%definite on $P(\mathcal{H}_{d})\times P(\mathcal{H}_{d})$. Therefore, there exists a
%subset $\overline{\mathcal{H}}_{\QJD}\subseteq\overline{\mathcal{H}}$ of a real separable Hilbert space $\overline{\mathcal{H}}$ and a
%one-to-one bijection $\Phi$ between $P(\mathcal{H}_{d})$ and $\overline
%{\mathcal{H}}_{\QJD}$ such that, for all $(\rho,\sigma)\in P(\mathcal{H}%
%_{d})\times P(\mathcal{H}_{d})$, $\QJD_{\alpha}(\rho,\sigma)=\Vert\Phi
%(\rho)-\Phi(\sigma)\Vert^{2}$ with $\Vert\cdot\Vert$ denoting the norm in
%$\overline{\mathcal{H}}$.
%\end{theorem}

\begin{lemma}
\label{qjsdpurensdlem} The distance space $\left(
P(\mathcal{H}_{d}),\QJD_{\alpha}\right), \alpha\in(0,2] $ is of negative type.
\end{lemma}

\begin{proof}
Using the same techniques as in Theorem \ref{theorem:2}, we have to prove that for $\rho\in P(\mathcal{H}_d)$, the function
$\rho \curvearrowright\mbox{\rm Tr}\big(\exp(
-t\rho)\big)  $ is exponentially convex.
For $\rho_1,\rho_2\in P(\mathcal{H}_d)$ such that $\rho_1\not =\rho_2$, the matrix
$\frac{\rho_1+ \rho_2}{2}$ has two non-zero eigenvalues, $\lambda_{+}$ and
$\lambda_{-}$, which can be calculated in the same way as above. In this case
(\ref{qeval}) reduces to
\[
\lambda_{\pm} = \frac{1}{2} \pm\frac{1}{2}\big( \mbox{\rm Tr}(\rho_1\cdot
\rho_2)\big) ^{1/2}.
\]
When we plug this into $\mbox{\rm Tr}\big(\exp(-t(\rho_1+ \rho_2))\big)$, we get
\beqrn
\mbox{\rm Tr}\left( e^{-2t\left( \frac{\rho_1+ \rho_2}{2}\right) }\right) &=& (n-2) \\
&& + 2e^{-t}\cosh\left( t\left( \mbox{\rm Tr}(\rho_1\cdot\rho_2)\right)^{1/2}\right) \\
&=& (n-2)\\
&& + 2e^{-t}\sum_{k=0}^{\infty}\frac{t^{2k}\left( \mbox{\rm Tr}(\rho_1
\cdot\rho_2)\right) ^{k}}{(2k)!},
\eeqrn
where the $(n-2)$ term comes from the fact that ${n-2}$ of the eigenvalues are
zero. We need to prove that $\left(  \rho_1,\rho_2\right)  \curvearrowright
\left( \mbox{\rm Tr}(\rho_1\cdot\rho_2)\right) ^{k}$ is positive definite
for all integers $k\geq0$. But Theorem~\ref{productfunctions} implies that we only need to
prove it for $k=1$. Appealing to the trace distance, we have
\[
\|\rho_1-\rho_2\|_{1}^{2} = \mbox{\rm Tr}\rho_1^{2} + \mbox{\rm Tr}\rho_2^{2} -
2\mbox{\rm Tr}(\rho_1\cdot\rho_2),
\]
Since, by Lemma \ref{Hmetricnsd}, this is negative definite, the result follows.
\end{proof}

\begin{proof}[ of Theorem \ref{qjsd:pure}]
Follows directly from Lemma \ref{qjsdpurensdlem} and Theorem \ref{theorem:1}.
\end{proof}

\subsection{Counter examples}

\subsubsection{Metric space counter example for $\alpha\in(2,3)$.}

To see that $\JD_{\alpha}$, and hence $\QJD_{\alpha}$, 
is not the square of a metric for
all~$\alpha$ we check the triangle inequality for the three probability
vectors $P=\left(  0,1\right)  ,Q=\left(  1/2,1/2\right)  $ and $R=\left(
1,0\right)  .$ We have
\begin{align*}
\JD_{\alpha}\left(  P,Q\right)    & =\JD_{\alpha}\left(  Q,R\right)  \\
& =S_{\alpha}\left(  1/4,3/4\right)  -\frac{S_{\alpha}\left(  1/2,1/2\right)
}{2}%
\end{align*}
and%
\[
\JD_{\alpha}\left(  P,R\right)  =S_{\alpha}\left(  1/2,1/2\right)  .
\]
The triangle inequality is equivalent to the inequality%
\begin{align*}
0 &  \geq-2\JD_{\alpha}\left(  P,Q\right)  -2\JD_{\alpha}\left(  Q,R\right)
+\JD_{\alpha}\left(  P,R\right)  \\
&  =-4\left(  S_{\alpha}\left(  1/4,3/4\right)  -\frac{S_{\alpha}\left(
1/2,1/2\right)  }{2}\right)  +S_{\alpha}\left(  1/2,1/2\right)  \\
&  =3S_{\alpha}\left(  1/2,1/2\right)  -4S_{\alpha}\left(  1/4,3/4\right)  \\
&  =3\frac{1-2\left(  1/2\right)  ^{\alpha}}{\alpha-1}-4\frac{1-\left(
1/4\right)  ^{\alpha}-\left(  3/4\right)  ^{\alpha}}{\alpha-1}\\
&  =\frac{4\left(  1/4\right)  ^{\alpha}+4\left(  3/4\right)  ^{\alpha
}-6\left(  1/2\right)  ^{\alpha}-1}{\alpha-1}.
\end{align*}
We make the substitution $x=\left(  1/2\right)  ^{\alpha}$ and assume
$\alpha>1$ so the inequality is equivalent to%
\[
4x^{2}+4x^{\frac{\ln4-\ln3}{\ln2}}-6x-1\leq0.
\]
Define the function
\[
f\left(  x\right)  =4x^{2}+4x^{2-\frac{\ln3}{\ln2}}-6x-1.
\]
Then its first and second derivatives are given by
\begin{align*}
f^{\prime}\left(  x\right)   &  =8x+4\left(  2-\frac{\ln3}{\ln2}\right)
x^{1-\frac{\ln3}{\ln2}}-6\\
f^{\prime\prime}\left(  x\right)   &  =8+4\left(  2-\frac{\ln3}{\ln2}\right)
\left(  1-\frac{\ln3}{\ln2}\right)  x^{-\frac{\ln3}{\ln2}}%
\end{align*}
and we see that $f^{\prime\prime}\left(  x\right)  =0$ has exactly one
solution. Therefore $f$ has exactly one infliction point and the equation
$f\left(  x\right)  =0$ has at most three solutions. Therefore the equation%
\[
4\left(  1/4\right)  ^{\alpha}+4\left(  3/4\right)  ^{\alpha}-6\left(
1/2\right)  ^{\alpha}-1=0
\]
has at most three solutions. It is straightforward to check that $\alpha=1,$
$\alpha=2$ and $\alpha=3$ are solutions, so these are the only ones. Therefore
the sign of%
\[
\frac{4\left(  1/4\right)  ^{\alpha}+4\left(  3/4\right)  ^{\alpha}-6\left(
1/2\right)  ^{\alpha}-1}{\alpha-1}%
\]
is constant in the interval $(2,3)$ and plugging in any number will show that
it is negative in this interval. Hence $\JD_{\alpha}$ cannot be a square of a
metric for $\alpha\in(2,3).$

\subsubsection{Counter examples for Hilbert space embeddability for $\alpha
\in\left(  \frac{7}{2},\infty\right)  .$}

In the previous paragraph we showed that $\JD_{\alpha}$ and $\QJD_{\alpha}$
are not the squares of metric functions for $\alpha\in(2,3)$. Hence, for
$\alpha$ in this interval, Hilbert space embeddings are not possible. Here we
prove a weaker result for $\alpha\in(\frac{7}{2},\infty)$, using the
Cayley-Menger determinant.

\begin{theorem}
The space 
\[
\left(  \mathcal{B}_{+}^{1}(\mathcal{H}_{d}), (\JD_{\alpha})^{\frac12}\right)  
\]
is \emph{not} Hilbert space embeddable for $\alpha$ in the interval $\left(
\frac{7}{2},\infty\right)  .$
\end{theorem}

Note that this does not exclude the possibility that $\JD_{\alpha}$ is the
square of a metric and that the same result holds for $\QJD_{\alpha},$

\begin{proof}
Consider the four distributions 
\begin{align*}
&\left(  \frac{1}{2}-3\varepsilon,\frac{1}{2}+3\varepsilon\right),\\
&\left(  \frac{1}{2}-\varepsilon,\frac{1}{2}+\varepsilon\right),\\
&\left(  \frac{1}{2}+\varepsilon,\frac{1}{2}-\varepsilon\right),\\
&\left(  \frac{1}{2}+3\varepsilon,\frac{1}{2}-3\varepsilon\right).
\end{align*}
 Then the
Cayley-Menger determinant is
\[
\left\vert
\begin{array}
[c]{ccccc}%
s_{\alpha}\left(  \frac{1}{2}-3\varepsilon\right)   & s_{\alpha}\left(
\frac{1}{2}-2\varepsilon\right)   & s_{\alpha}\left(  \frac{1}{2}%
-\varepsilon\right)   & s_{\alpha}\left(  \frac{1}{2}\right)   & 1\\
s_{\alpha}\left(  \frac{1}{2}-2\varepsilon\right)   & s_{\alpha}\left(
\frac{1}{2}-\varepsilon\right)   & s_{\alpha}\left(  \frac{1}{2}\right)   &
s_{\alpha}\left(  \frac{1}{2}+\varepsilon\right)   & 1\\
s_{\alpha}\left(  \frac{1}{2}-\varepsilon\right)   & s_{\alpha}\left(
\frac{1}{2}\right)   & s_{\alpha}\left(  \frac{1}{2}+\varepsilon\right)   &
s_{\alpha,2}\left(  \frac{1}{2}+2\varepsilon\right)   & 1\\
s_{\alpha}\left(  \frac{1}{2}\right)   & s_{\alpha}\left(  \frac{1}%
{2}+\varepsilon\right)   & s_{\alpha}\left(  \frac{1}{2}+2\varepsilon\right)
& s_{\alpha}\left(  \frac{1}{2}+3\varepsilon\right)   & 1\\
1 & 1 & 1 & 1 & 0
\end{array}
\right\vert
\]
and if the four points are Hilbert space embeddable then this determinant is
non-negative. The function $\varepsilon\rightarrow s_{\alpha}\left(  \frac
{1}{2}+\varepsilon\right)  $ has a Taylor expansion given by
\begin{multline}
s_{\alpha}\left(  \frac{1}{2}+\varepsilon\right)  =s_{\alpha}\left(  \frac
{1}{2}\right)  +\frac{s_{\alpha}^{\prime\prime}\left(  \frac{1}{2}\right)
}{2}\varepsilon^{2} +\\
\frac{s_{\alpha}^{\left(  4\right)  }\left(  \frac{1}
{2}\right)  }{24}\varepsilon^{4}+ \frac{s_{\alpha}^{\left(  6\right)  }\left(
\frac{1}{2}\right)  }{720}\varepsilon^{4}+\varepsilon^{8}f\left(
\varepsilon\right),
\end{multline}
where $f$ is some continous function of $\varepsilon.$ This can be used to get
the expansion of the Cayley-Menger determinant:
\begin{multline*}
\CM=\frac{1}{8}s_{\alpha}^{\left(  4\right)  }\left(  \frac{1}{2}\right)
\left(  \left(  s_{\alpha}^{\left(  4\right)  }\left(  \frac{1}{2}\right)
\right.\right)  ^{2}-\\
\left. s_{\alpha}^{\prime\prime}\left(  \frac{1}{2}\right)  h_{\alpha
}^{\left(  6\right)  }\left(  \frac{1}{2}\right)  \right)  \varepsilon
^{12}+\varepsilon^{14}g\left(  \varepsilon\right)
\end{multline*}
for some continuous function $g$ \footnote{The calculation of the Taylor expansion involves a lot of computations but are easily performed using Maple or similar symbol manipulation program.}. We have the following
formula for the even derivatives of $s_{\alpha}:$%
\[
s_{\alpha}^{\left(  2n\right)  }\left(  x\right)  =-\alpha^{\underline{2n}%
}\left(  x^{\alpha-2n}+\left(  1-x\right)  ^{\alpha-2n}\right)
\]
and
\[
s_{\alpha}^{\left(  2n\right)  }\left(  \frac{1}{2}\right)  =-\alpha
^{\underline{2n}}2^{2n+1-\alpha}.
\]
If the Cayley-Menger determinant is positive for all small~$\varepsilon$ then
\[
\left(  s_{\alpha}^{\left(  4\right)  }\left(  \frac{1}{2}\right)  \right)
^{2}-s_{\alpha}^{\prime\prime}\left(  \frac{1}{2}\right)  s_{\alpha}^{\left(
6\right)  }\left(  \frac{1}{2}\right)  \leq0
\]
or equivalently%
\[
\left(  -\alpha^{\underline{4}}2^{5-\alpha}\right)  ^{2}-\left(
-\alpha^{\underline{2}}2^{3-\alpha}\right)  \left(  -\alpha^{\underline{6}%
}2^{7-\alpha}\right)  \leq0
\]
and
\begin{align*}
0 &  \geq\left(  \alpha^{\underline{4}}\right)  ^{2}-\left(  \alpha
^{\underline{2}}\right)  \left(  \alpha^{\underline{6}}\right)  \\
&  =\alpha^{\underline{2}}\alpha^{\underline{4}}\left(  \left(  \alpha
-2\right)  \left(  \alpha-3\right)  -\left(  \alpha-4\right)  \left(
\alpha-5\right)  \right)  \\
%&  =\alpha^{\underline{2}}\alpha^{\underline{4}}\left(  4\alpha-14\right)  \\
&  =4\alpha^{\underline{2}}\left(  \alpha-2\right)  \left(  \alpha-3\right)
\left(  \alpha-\frac{7}{2}\right)  .
\end{align*}
Hence, the Cayley-Menger determinant is non-negative only for the intervals $\left[0, 2\right]  $ and $\left[  3, \frac{7}%
{2}\right]  .$
\end{proof}

\section{Relation to total variation and trace distance}\label{SecUpDown}

The results of Section \ref{metricsec} indicate that interesting geometric
properties are associated with $\JD_{\alpha}$ and $\QJD_{\alpha}$ when
$\alpha\in\left(  0,2\right]  .$

\subsection{Bounds on $\JD_{\alpha}$}

\label{vjsd}

For $\alpha\in(0,2].$ we bound $\JD_{\alpha}$ as follows:

%\begin{align}
%U_{n}(v) &  =\sup\{\JD(P,Q)\mid P,Q\in M_{+}^{1}%
%(n),V(P,Q)=v\}\,,\label{eq:18}\\
%L_{n}(v) &  =\inf\{\JD(P,Q)\mid P,Q\in M_{+}^{1}(n),V(P,Q)=v\}\,,
%\end{align}
%for the upper and lower bound, respectively.

\begin{theorem}\label{lowerupperthm}
Let $P$ and $Q$ be probability distributions in $M_{+}^{1}(n),$ and let
\[
v:=\tfrac{1}{2}\sum_{i}\left\vert p_{i}-q_{i}\right\vert \in\lbrack0,2]
\]
denote their total variation. Then for $\alpha\in\left(  0,2\right]  ,$ we
have $L\leq\JD_{\alpha}(P,Q)\leq U,$ where:

\begin{itemize}
\item For every $n\geq2,$ $L$ is given by
\begin{equation}
L(P,Q)=s_{\alpha}\left(  \tfrac{1}{2}\right)  -s_{\alpha}\left(  \tfrac
{1}{2}+\tfrac{v}{4}\right)  . \label{eq:27}%
\end{equation}
%where $s_{\alpha}(p) = S_{\alpha}(p,1-p)$ is the binary entropy function.

\item For every $n\geq3,$ $U$ is given by
\begin{equation}
U_{n}(P,Q)=\frac{1}{\alpha-1}\left(  \frac{1}{2}-\frac{1}{2^{\alpha}}\right)
\Vert P-Q\Vert_{\alpha}^{\alpha}. \label{eq:28}%
\end{equation}

\item For $n=2,$ $U$ is given by the tighter quantity
\begin{equation}
U_{2}(P,Q)=s_{\alpha}\left(  \tfrac{v}{4}\right)  -\frac{1}{2}S_{\alpha
,2}\left(  \tfrac{v}{2}\right)  . \label{eq:29}%
\end{equation}

\end{itemize}
\end{theorem}

\begin{proof}
We start with the lower bound. Let $\sigma$ denote a permutation of the elements in $\left[  n\right]  $ and let
$\sigma\left(  P\right)  $ denote the probability vector where the point
probabilities have been permuted according to $\sigma$. Clearly, the function $\JD_{\alpha}$ is invariant under such permutations of its arguments:
\beq\label{JDpermuted}
\JD_{\alpha}\big(\sigma(P)  ,\sigma(Q)\big)
=\JD_{\alpha}\left(  P,Q\right).
\eeq
Let $B$ denote the set of permutations $\sigma$ that satisfy
\[
p_{i}\geq q_{i}\Leftrightarrow p_{\sigma\left(  i\right)  }\geq q_{\sigma\left(
i\right)  }%
\]
for all $i\in\left[  n\right].$ Then, by the joint convexity of $\JD_{\alpha}$ for $\alpha\in[1,2]$ (as proved in \cite{Burbea1982}), we have%
\begin{align}
\JD_{\alpha}\left(  P,Q\right)   &  =\frac{1}{\left\vert B\right\vert }%
\sum_{\sigma\in B}\JD_{\alpha}\big(\sigma(P)  ,\sigma(Q)\big)\nonumber  \\
&  \geq\JD_{\alpha}\left(  \frac{1}{\left\vert B\right\vert }\sum_{\sigma\in
B}\sigma\left(  P\right)  ,\frac{1}{\left\vert B\right\vert }\sum_{\sigma\in B}%
\sigma\left(  Q\right)  \right)  .\label{jointconv}
\end{align}
The distributions $\frac{1}{\left\vert B\right\vert }\sum_{\sigma\in B}\sigma\left(
P\right)  $ and $\frac{1}{\left\vert B\right\vert }\sum_{\sigma\in B}\sigma\left(
Q\right)  $ have the property that they are constant on two complementary
sets, namely $\{i\in[n]\mid p_i\geq q_i\}$ and $\{i\in[n]\mid p_i<q_i\}$. Therefore, we may without loss of generality assume that $P$ and $Q$ are
distributions on a two-element set. On a two-element set $P$ and $Q$ can be
parametrized by $P=\left(  p,1-p\right)  $ and $Q=\left(  q,1-q\right).$ If
$\sigma_2$ denotes the transposition of the two elements then
\begin{align*}
v &  =V\left(  \frac{P+\sigma_2\left(  Q\right)  }{2},\frac{Q+\sigma_2\left(  P\right)}{2}\right)  =2\left\vert p-q\right\vert .
\end{align*}
% We also have%
% \[
% \JD_{\alpha}\left(  P,Q\right)  =\JD_{\alpha}\big(\sigma(Q),\sigma(P)\big)  .
% \]
By \eqref{JDpermuted} and \eqref{jointconv}  we get%
\begin{multline*}
\JD_{\alpha}\left(  P,Q\right)     \geq\JD_{\alpha}\left(  \frac{P+\sigma_2\left(
Q\right)  }{2},\frac{Q+\sigma_2\left(  P\right)  }{2}\right)  \\
%&  =\JD_{\alpha}\left(  \left(  \frac{1}{2}+\frac{p-q}{2},\frac{1}{2} -\frac{p-q}{2}\right)  ,\left(  \frac{1}{2}-\frac{p-q}{2},\frac{1}{2} +\frac{p-q}{2}\right)  \right)  \\
  =\JD_{\alpha}\left(  \left(  \frac{1}{2}+\frac{v}{4},\frac{1}{2}-\frac
{v}{4}\right)  ,\left(  \frac{1}{2}-\frac{v}{4},\frac{1}{2}+\frac{v}%
{4}\right)  \right)  \\
  =s_{\alpha}\left(  1/2\right)  -s_{\alpha}\left(  \frac{1}{2}+\frac
{v}{4}\right),
\end{multline*}
and this lower bound is attained for two distributions on a two element set.
Next we derive the general upper bound. Define distribution $\widetilde{P}$ on $[n]\times [3]$ such that for every $i\in[n]$,
\begin{align*}
\widetilde{P}\left(i,1\right)    & =\min\left\{  p_i  ,q_i  \right\}  ,\\
\widetilde{P}\left(  i,2\right)    & =\left\{
\begin{array}
[c]{ll}
p_i-q_i  & \text{if }p_i>q_i\\
0 & \text{otherwise},
\end{array}
\right. \\
\widetilde{P}\left(i,3\right)    & =0,
\end{align*}
and similarly define $\widetilde{Q}$ on $[n]\times [3]$ by
\begin{align*}
\widetilde{Q}\left(i,1\right)    & =\min\left\{ p_i.q_i  \right\}  ,\\
\widetilde{Q}\left(i,2\right)    & =0,\\
\widetilde{Q}\left(i,3\right)    & =\left\{
\begin{array}
[c]{ll}%
q_i-p_i & \text{if }q_i>p_i\\
0 & \text{otherwise}.
\end{array}
\right.
\end{align*}
With these definitions we have $V(\widetilde{P},\widetilde{Q})  =V\left(
P,Q\right)  .$ Using the data processing inequality and the definitions of $\widetilde{P}$ and $\widetilde{Q}$ it is straighforward to verify that
\beqrn
\JD_{\alpha}\left(  P,Q\right)    &\leq& \JD_{\alpha}(\widetilde{P},\widetilde{Q})\\
&=& \frac{1}{\alpha-1}\big(\frac{1}{2} - \frac{1}{2^{\alpha}}\big)\sum_{i=1}^n|p_i-q_i|^{\alpha}.
\eeqrn
This upper bound is attained on a three element set so we have
\[
U_n(P,Q) =\frac{1}{\alpha-1}\big(\frac{1}{2} - \frac{1}{2^{\alpha}}\big)\|P-Q\|_{\alpha}^{\alpha}.
\]
To get a tight upper bound on a two-element set a special analysis is needed. The cases $p>q$ and $p<q$ are treated separately, but the two cases work the same way. We will therefore assume that $p>q.$
On a two-element set parametrize $P$ and $Q$ by $P=\left(
p,1-p\right)  $ and $Q=\left(  q,1-q\right)  .$ In this case we have the linear
constraint $p-q=v/2$. For a
fixed value of $v$, we have that $\JD_{\alpha}$ is a convex function of $q.$
Therefore the maximum is attained by an extreme point, i.e. a distribution
where either $p$ or $q$ is either $0$ or $1.$ Without loss of generality we
may assume that $q=0$ and that $p=v/2.$ This gives%
\[
U_2(P,Q) = s_{\alpha}\left(  \frac{v}{4}\right)  -\frac{s_{\alpha}\left(  \frac{v}{2}\right)  }{2}.\qedhere
\]
% As a conclusion we have the following three cases:%
% \begin{align}
% &  \mbox{Case {\it a}:}\qquad\JD(P,Q)=-\frac{1}{2}H_2\Big(1 + \frac{v}{2}\Big)\label{eq:25a}\\
% &  \mbox{Case {\it b}:}\qquad\JD(P,Q)=\ln2-\frac{1}{2}\big(2-\frac{v}{2}\big)\ln\big(2-\frac{V}{2}\big)+\frac{1}{2}\big(1-\frac{v}{2}\big)\ln\big(1-\frac{v}{2}\big)\label{eq:25b}\\
% &  \mbox{Case {\it c}:}\qquad\JD(P,Q)=\frac{\ln2}{2}\cdot v\,.\label{eq:25c}
% \end{align}
% The functions in \eqref{eq:25a}-\eqref{eq:25c} are related as follows:
% \begin{multline}
% -\frac{1}{2}H_2\Big(1 + \frac{v}{2}\Big)< \ln2-\frac{1}{2}\big(2-\frac{v}{2}\big)\ln\big(2-\frac{v}{2}\big)+\frac
% {1}{2}\big(1-\frac{v}{2}\big)\ln\big(1-\frac{v}{2}\big)< \frac{\ln2}{2}\cdot v\label{eq:26}
% \end{multline}
% This follows by standard considerations. Of course, for $v=0$ or $v=2$,
% equality holds in \eqref{eq:26}. Recalling that \eqref{eq:25c} requires
% $n\geq3$ gives the result.
\end{proof}

It is now straightforward to determine the exact form of the joint range of
$V$ and $\JD_{\alpha}$.

\begin{cor}
The joint range of $V$ and $\JD_{\alpha},$ denoted by $\Delta_{n},$ is a
compact region in the plane bounded by a (Jordan) curve composed of two
curves: The first curve is given by \eqref{eq:27} with $V$ running from 2 to
0. For $n=2$ the second curve is given by \eqref{eq:28} with $v$ running from
0 to 2, and for $n=3$ the second curve is given by \eqref{eq:29} with $v$
running from 0 to 2.
\end{cor}

\begin{proof}
Assume first that $n\geq3$. By Theorem \ref{lowerupperthm} we know that $\Delta_{n}$ is
contained in the compact domain described. A continuous deformation of the lower
curve into the upper bounding curve (i.e. a homotopy  from the lower bounding
curve to the upper bounding curve ) is given by $P_{t}$, $Q_{t}$ for $t\in\lbrack0,1]$,
where
\beqrn
\begin{pmatrix}
P_{t}\\
Q_{t}
\end{pmatrix}
(v)&=& (1-t)
\begin{pmatrix}
\frac{2+v}{4} & \frac{2-v}{4} & 0 & \cdots & 0\\
\frac{2-v}{4} & \frac{2+v}{4} & 0 & \cdots & 0
\end{pmatrix}
+\\
&&t
\begin{pmatrix}
1-\frac{v}{2} & \frac{v}{2} & 0 &  0 & \cdots & 0\\
1-\frac{v}{2} & 0 & \frac{v}{2} &  0 & \cdots & 0
\end{pmatrix}
\eeqrn
for $v\in[0,2].$ Therefore, $\Delta_{n}$ has no \textquotedblleft
holes\textquotedblright. The case $n=2$ is handled in a similar way.
\end{proof}

\begin{figure}[h]
\begin{center}
\includegraphics[width=.8\columnwidth]{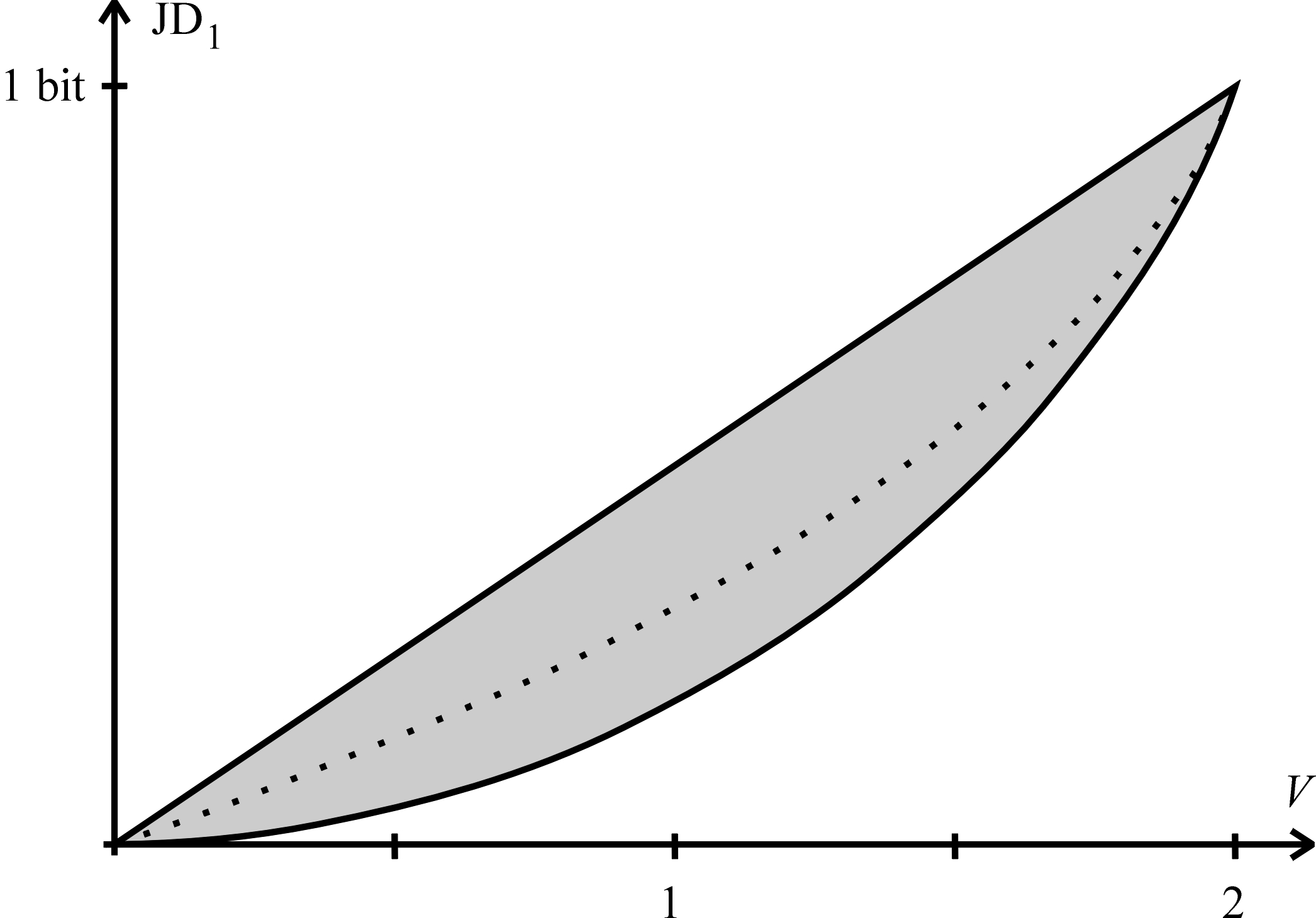}
\end{center}
\caption{$V/\JD_{\alpha}$-diagram for $\alpha=1$ and $n\geq3$ (the shaded
region), and for $n=2$ (the region obtained by replacing the upper bounding
curve by the dotted curve).}%
\label{vjsddiagram}%
\end{figure}

In Figure \ref{vjsddiagram} we have depicted the $V/\JD_{\alpha}$-diagram for
$\alpha=1$.

The bounds \eqref{eq:27} and \eqref{eq:28} give us the following proposition
regarding the topology induced by $(\JD_{\alpha})^{\frac12}$.
In the limiting case $\alpha\to1$,
this was proved in \cite{Topsoe00ine} by a different method.

\begin{proposition}
\label{theorem:3} The space $\left(  M_{+}^{1}(\mathbb{N}),\JD_{\alpha}%
^{1/2}\right)  $ is a complete, bounded metric space for $\alpha\in\left(
0,2\right]  $,  and the induced topology is that of convergence in total variation.
\end{proposition}

\begin{proof}By expansion of $L(P,Q)$ given by \eqref{eq:27}, in terms of the total variation $v$, one obtains the inequality
\begin{equation}
\JD_{\alpha}(P,Q)\geq\frac{1}{\alpha-1}\sum_{j=1}^{\infty}\binom{\alpha}{2j}\Big(\frac{v}{2}\Big)^{2j}.
\end{equation}
Taking only the first term and bounding \eqref{eq:28}, we get
\beqr
\frac{1}{8}V^2(P,Q) &\leq& \frac{\alpha}{8}V^{2}(P,Q)\nonumber\\
&\leq& \JD_{\alpha}(P,Q)\nonumber\\
&\leq& \frac{1}{\alpha-1}\big(\frac{1}{2} - \frac{1}{2^{\alpha}}\big)\|P-Q\|_{\alpha}^{\alpha}\nonumber\\
&\leq& \frac{\ln 2}{2}V(P,Q)\,.\qedhere\label{varbounds}
\eeqr
\end{proof}

\subsection{Bounds on $\QJD_{\alpha}$}

With Theorem \ref{lowerupperthm} we can bound $\QJD_{\alpha}$ for $\alpha
\in[1,2]$. We use the following two theorems.

\begin{theorem}
[\cite{Petz2008}, Theorem 3.9]\label{monotonicity} Let $\mathcal{H}$ be a
Hilbert space, $\rho_{1},\rho_{2}\in\mathcal{B}_{+}^{1}\left(  \mathcal{H}%
\right)  $ and $\mathcal{M}:=\{M_{i}\mid i=1,\dots,\,n\}$ be a measurement on
$\mathcal{H}$. Then $S(\rho_{1}\Vert\rho_{2})\geq D(P_{\mathcal{M}}\Vert
Q_{\mathcal{M}})$, where $P_{\mathcal{M}},Q_{\mathcal{M}}\in M_{+}^{1}(n)$ and
have point probabilities $P_{\mathcal{M}}(i)=\mbox{\rm Tr}(M_{i}\rho_{1})$ and
$Q_{\mathcal{M}}(i)=\mbox{\rm Tr}(M_{i}\rho_{2})$, respectively.
\end{theorem}

\begin{theorem}
[\cite{Nielsen2000}, Theorem 9.1]\label{tracevariance} Let $\mathcal{H}$ be a
Hilbert space, 
\[
\rho_{1},\rho_{2}\in\mathcal{B}_{+}^{1}\left(  \mathcal{H}%
\right)  
\]
 and $\mathcal{M}:=\{M_{i}\mid i=1,\dots,\,n\}$ be a measurement on
$\mathcal{H}$. Then $\Vert\rho_{1}-\rho_{2}\Vert_{1}=\max_{\mathcal{M}%
}V(P_{\mathcal{M}},Q_{\mathcal{M}})$, where $P_{\mathcal{M}},Q_{\mathcal{M}%
}\in M_{+}^{1}(n)$ and have point probabilities $P_{\mathcal{M}}%
(i)=\mbox{\rm Tr}(M_{i}\rho_{1})$ and $Q_{\mathcal{M}}(i)=\mbox{\rm Tr}(M_{i}%
\rho_{2})$, respectively.
\end{theorem}

\begin{theorem}
For $\alpha\in(0,2]$, for all states $\rho_{1},\rho_{2}\in
\mathcal{B}_{+}^{1}\left(  \mathcal{H}\right)  $, we have
\begin{align*}
s_{\alpha}(\tfrac{1}{2})-s_{\alpha}\left(  \frac{1}{2}+\frac{\Vert\rho
_{1}-\rho_{2}\Vert_{1}}{2}\right)   &  \leq\QJD_{\alpha}(\rho_{1},\rho_{2})\\
&  \leq\frac{\ln2}{2}\Vert\rho_{1}-\rho_{2}\Vert_{1}.
\end{align*}

\end{theorem}

\begin{proof}
The lower bound is proved in the same way as \cite[Theorem III.1]%
{KNTZ:interaction}, by making a reduction to the case of classical probability
distributions by means of measurements. Let $\mathcal{M}$ be a measurement
that maximizes $V(P_{\mathcal{M}},Q_{\mathcal{M}})$. Then from Theorem
\ref{tracevariance} we have $\Vert\rho_1-\rho_2\Vert_{1}=V(P_{\mathcal{M}%
},Q_{\mathcal{M}})$. Theorem \ref{monotonicity} gives us
\beqrn
\QJD_{\alpha}(\rho_1,\rho_2) &\geq& \frac{1}{2}D\left(  P_{\mathcal{M}}\big\|\frac{P_{\mathcal{M}}+Q_{\mathcal{M}}}{2}\right)\\
&& +\frac{1}{2}D\left(Q_{\mathcal{M}}\big\|\frac{P_{\mathcal{M}}+Q_{\mathcal{M}}}{2}\right)\\
&=& \JD_{\alpha}(P_{\mathcal{M}},Q_{\mathcal{M}}).
\eeqrn
The result now follows from Theorem \ref{lowerupperthm}.
The upper bound is proved the same way as we proved the classical bound.
Introduce a 3-dimensional Hilbert space $\mathcal{G}$ with basis vectors
$\ket{1}$, $\ket{2}$ and $\ket{3}$. On $\mathcal{H\otimes G}$ define the density matrices%
\beqrn
\tilde{\rho}_1  &=& \frac{\rho_1 + \rho_2 - |\rho_1 - \rho_2| }{2}\otimes\ketbra{1}{1}\\
&& +\frac{\rho_1 - \rho_2 + |\rho_1 - \rho_2| }{2}\otimes\ketbra{2}{2} ,\\
\tilde{\rho}_2  &=& \frac{\rho_2 + \rho_1 - |\rho_2 - \rho_1|}{2}\otimes\ketbra{1}{1}\\
&& +\frac{\rho_2 - \rho_1 + |\rho_1 - \rho_2| }{2}\otimes\ketbra{3}{3} .
\eeqrn
Let $\Tr_{\mathcal{G}}$ denote the partial trace $\mathcal{B}_{+}^{1}(\mathcal{H}%
\otimes\mathcal{G})\rightarrow\mathcal{B}_{+}^{1}(\mathcal{H}).$ Then
$\Tr_{\mathcal{G}}\left(  \tilde{\rho}_1\right)  =\rho_1$ and $\Tr_{\mathcal{G}}\left(
\tilde{\rho}_2\right)=\rho_2  .$ The matrices $\frac{\rho_1 - \rho_2 + |\rho_1 - \rho_2|}{2}$ and $\frac{\rho_2 - \rho_1 + |\rho_1 - \rho_2|}{2}$ are positive definite so%
\beqrn
\Vert\tilde{\rho}_1-\tilde{\rho}_2\Vert_{1}  &=& \Tr\left\vert \frac{\rho_1 - \rho_2 + |\rho_1 - \rho_2| }{2}\otimes\ketbra{2}{2}\right.\\
&& \left.-\frac{\rho_2 - \rho_1 + |\rho_1 - \rho_2|}{2}\otimes\ketbra{3}{3} \right\vert \\
&=& \Tr\left(  \frac{\rho_1 - \rho_2 + |\rho_1 - \rho_2| }{2}\right)\\
&& +\Tr\left(\frac{\rho_2 - \rho_1 + |\rho_1 - \rho_2| }{2}\right)  \\
&=& \Tr\left\vert \rho_1-\rho_2\right\vert =\Vert\rho_1-\rho_2\Vert_{1}.
\eeqrn
According to the ``quantum data processing inequality'' \cite[Theorem 3.10]{Petz2008} we have%
\begin{multline*}
\QJD_{\alpha}(\rho_1,\rho_2)  \leq \QJD_{1}(\tilde{\rho}_1,\tilde{\rho}_2)\\
= \frac{1}{2}\Tr\left(  \frac{\rho_1 - \rho_2 + |\rho_1 - \rho_2|
}{2}\otimes\ketbra{2}{2}\right)\ln2\\
 + \Tr\left(\frac{\rho_2 - \rho_1 + |\rho_1 - \rho_2|
}{2} \otimes\ketbra{3}{3}\right)
\ln2\\
= \frac{\ln2}{2}\cdot\Vert\rho_1-\rho_2\Vert_{1}\,.\qedhere
\end{multline*}
\end{proof}

\section{Conclusions and open problems}

We studied generalizations of the (general) Jensen divergence and its
quantum analogue. For $\alpha\in\left(  1,2\right]  $, $\JD_{\alpha}$ was
proved to be the square of a metric which can be embedded in a real Hilbert
space. The same was shown to hold for $\QJD_{\alpha}$ restricted to qubit
states or to pure states. Both these results were derived by evoking a theorem
of Schoenberg's and showing that these quantities are negative definite.

Whether $(\QJD_{1})^{\frac12}$ is a metric for all mixed states remains unknown. However,
based on a large amount of numerical evidence, we conjecture the function
$A\rightarrow\mbox{\rm Tr}(e^{A})$ to be exponentially convex for density
matrices $A$. Proving this would imply that $\QJD_{\alpha}$ is negative
definite for $\alpha\in(0,2]$, and hence the square of a metric that
can be embedded in a real Hilbert space.

\section{Acknowledgements}
We are greatly indebted to Flemming Tops{\o}e. This work mainly extends his basic result jointly with Bent Fuglede presented at the conference ISIT 2004 \cite{Fuglede2004}, where you find the basic result on isometric embedding in Hilbert space related to $\JD_1$. Flemming has supplied us with many valuable comments and suggestions. In particular Section \ref{SecUpDown} is to a large extent inspired by unpublished results of Flemming.

Jop Bri{\"e}t is partially supported by a Vici grant from the Netherlands Organization for
Scientific Research (NWO), and by the European Commission under the Integrated
Project Qubit Applications (QAP) funded by the IST directorate as Contract
Number 015848. Peter Harremo{\"e}s has been supported by the Villum Kann Rasmussen Foundation,
by Danish Natural Science Research Council, by INTAS (project 00-738) and by
the European Pascal Network. 

% \bibliographystyle{unsrt}
% \bibliography{database}

\end{document}